\documentclass[aps,prl,reprint,groupedaddress,showpacs,amsmath,nofootinbib]{revtex4-1}

\usepackage{amsthm}
\usepackage{amsfonts}
\usepackage{hyperref}

\newtheorem{theorem}{Theorem}
\newtheorem{proposition}[theorem]{Proposition}
\newtheorem{corollary}[theorem]{Corollary}

\newcommand{\Ket}[1]{\left \lvert #1\right \rangle}

\newcommand{\BraKet}[2]{\left \langle #1 \middle \vert #2\right \rangle}

\newcommand{\Abs}[1]{\left\lvert#1\right\rvert}

\newcommand{\QProb}[2]{\Abs{\BraKet{#1}{#2}}^2}

\begin{document}

\title{$\psi$-epistemic models are exponentially bad at explaining the
  distinguishability of quantum states}

\author{M. S. Leifer}
\email{matt@mattleifer.info}
\homepage{http://mattleifer.info}
\affiliation{Perimeter Institute for Theoretical Physics, 31 Caroline
  St. N, Waterloo, ON, Canada N2L 2Y5}

\date{March 26, 2014}

\begin{abstract}
  The status of the quantum state is perhaps the most controversial
  issue in the foundations of quantum theory.  Is it an
  \emph{epistemic state} (state of knowledge) or an \emph{ontic state}
  (state of reality)?  In realist models of quantum theory, the
  epistemic view asserts that nonorthogonal quantum states correspond
  to overlapping probability measures over the true ontic states.
  This naturally accounts for a large number of otherwise puzzling
  quantum phenomena.  For example, the indistinguishability of
  nonorthogonal states is explained by the fact that the ontic state
  sometimes lies in the overlap region, in which case there is nothing
  in reality that could distinguish the two states.  For this to work,
  the amount of overlap of the probability measures should be
  comparable to the indistinguishability of the quantum states.  In
  this letter, I exhibit a family of states for which the ratio of
  these two quantities must be $\leq 2de^{-cd}$ in Hilbert spaces of
  dimension $d$ that are divisible by $4$.  This implies that, for
  large Hilbert space dimension, the epistemic explanation of
  indistinguishability becomes implausible at an exponential rate as
  the Hilbert space dimension increases.
\end{abstract}

\pacs{03.65.Ta, 03.65.Ud}

\maketitle

The status of the quantum state is one of the most controversial
issues in the foundations of quantum theory.  Is it a state of
knowledge (an \emph{epistemic} state), or a state of physical reality
(an \emph{ontic state})?  Many realist interpretations of quantum
theory employ ontic quantum states, but the rise of quantum
information science has revived interest in the epistemic alternative
because many of the puzzling phenomena employed in quantum information
are explained quite naturally in terms of the epistemic quantum states
\cite{Spekkens2007}.  For example, consider the fact that two
nonorthogonal quantum states cannot be perfectly distinguished.  On
the ontic view, the two states represent distinct arrangements of
physical reality, so it is puzzling that this distinctness cannot be
detected.  However, on the epistemic view, a quantum state is
represented by a probability measure over the physical properties of a
system and nonorthogonal states correspond to overlapping probability
measures.  Indistinguishability is thus explained by the fact that
preparations of the two quantum states sometimes result in the same
physical properties of the system, in which case there is nothing
existing in reality that could distinguish them.

Recently, several theorems have been proved aiming to show that the
quantum state must be ontic \cite{Pusey2012, Hardy2013, Patra2013a,
  Colbeck2012, Colbeck2013}.  These have been proved within the
\emph{ontological models} framework \cite{Harrigan2010}, which is a
refinement of the hidden variable approach used to prove earlier no-go
results, such as Bell's theorem \cite{Bell1964} and the Kochen-Specker
theorem \cite{Kochen1967}.  Each of these theorems employs
questionable auxiliary assumptions\footnote{See \cite{Leifer2014} for
  a review of these theorems and the criticisms of them.}.  Without
such assumptions, explicit counterexamples show that the epistemic
view of quantum states can be maintained \cite{Lewis2012,
  Aaronson2013}.

Within the ontological models framework, a model is
\emph{$\psi$-ontic} if the probability measures corresponding to every
pair of pure quantum states have zero overlap, and it is
\emph{$\psi$-epistemic} otherwise.  The recent no go theorems aim to
show that models must be $\psi$-ontic and the counterexamples show
that, without auxiliary assumptions, $\psi$-epistemic models exist.
However, being $\psi$-epistemic is an extremely permissive notion of
what it means for the quantum state to be epistemic, since any nonzero
amount of overlap between probability measures, however small, is
enough to make a model $\psi$-epistemic.  On the other hand, the
$\psi$-epistemic explanation of indistinguishability requires that a
significant part of the indistinguishability of $\Ket{\psi}$ and
$\Ket{\phi}$ should be accounted for by the overlap of the
corresponding probability measures, which means that the overlap
should be comparable to a quantitative measure of the
indistinguishability of $\Ket{\psi}$ and $\Ket{\phi}$.  For this
reason, it is interesting to bound the overlaps in ontological models,
since this can be done without auxiliary assumptions, and it may still
render $\psi$-epistemic explanations implausible.

Along these lines, Maroney showed that a measure of the overlap of
probability measures must be smaller than $\QProb{\phi}{\psi}$ for
some pairs of states in systems with Hilbert space of dimension $d
\geq 3$ \cite{Maroney2012}, which was later shown to follow from
Kochen-Specker contextuality \cite{Maroney2012a, Leifer2013c}.
Following this, Barrett et al.\ showed that the ratio of an overlap
measure derived from the variational distance to a comparable measure
of the indistinguishability of quantum states must scale like
$4/(d-1)$ in Hilbert space dimension for a particular family of states
\cite{Barrett2013}.  In this letter, I exhibit a family of states in
Hilbert space dimensions $d$ that are divisible by $4$ for which the
same ratio must be $\leq de^{-cd}$, where $c$ is a positive constant.
Hence, for large Hilbert space dimension, the $\psi$-epistemic
explanation of indistinguishability becomes increasingly implausible
at an exponential rate as the Hilbert space dimension increases.

We are interested in ontological models that reproduce the quantum
predictions for prepare-and-measure experiments made on a system with
Hilbert space $\mathbb{C}^d$, with no time evolution between
preparation and measurement.

An ontological model for this set of experiments is an attempt to
explain the quantum predictions in terms of some real physical
properties---denoted $\lambda$ and called \emph{ontic states}---that
exist independently of the experimenter.  The set of ontic states is
denoted $\Lambda$ and we assume that it is a measurable space
$(\Lambda, \Sigma)$ with $\sigma$-algebra $\Sigma$.  When the
experimenter prepares a state $\Ket{\psi}$, the preparation device
might not fully control the ontic state, so $\Ket{\psi}$ is associated
with a probability measure $\mu_{\psi}:\Sigma \rightarrow
[0,1]$\footnote{In general, the probability measure should be
  associated with the method of preparing $\Ket{\psi}$ rather than
  with $\Ket{\psi}$ itself because different preparation procedures
  for the same state may result in different probability measures.
  This phenomenon is known as \emph{preparation contextuality} and it
  must occur for mixed states \cite{Spekkens2005}.  This complication
  does not affect the results presented here as the bounds derived
  apply equally well to any of the measures that can represent a
  preparation of $\Ket{\psi}$.}.

Similarly, measurements might not reveal the value of $\lambda$
exactly, so each element $\Ket{a}$ of an orthonormal basis $M = \left
  \{ \Ket{a}, \Ket{b}, \ldots \right \}$ is associated with a positive
measurable response function $\xi_M(a|\lambda)$, where
$\xi_M(a|\lambda)$ is the probability of obtaining the outcome
$\Ket{a}$ when the measurement in the basis $M$ is performed on the
system and the ontic state is $\lambda$.  Note that the response
functions are allowed to depend on $M$ to account for contextuality.
Additionally, in order to form a well-defined probability distribution
over the measurement outcomes, the response functions must satisfy
\begin{equation}
  \label{eq:response}
  \sum_{\Ket{a} \in M} \xi_M(a|\lambda) = 1. 
\end{equation}
Finally, the ontological model is required to reproduce the quantum
predictions, which means that, for all pure states $\Ket{\psi}$, all
orthonormal bases $M$, and all $\Ket{a} \in M$,
\begin{equation}
  \label{eq:reproduce}
  \int_{\Lambda} \xi_M(a|\lambda) d\mu_{\psi} = \QProb{a}{\psi}.
\end{equation}

It will prove useful to define the sets of ontic states
\begin{equation}
  \Gamma^{a}_M = \left \{ \lambda \middle | \xi_M(a|\lambda)
    = 1 \right \},
\end{equation}
which always yield the outcome $\Ket{a}$ with certainty when the
measurement $M$ is performed, and to note that, by
Eq.~\eqref{eq:response}, $\Gamma^{a}_M$ and $\Gamma^{b}_M$ are
disjoint for $\Ket{a} \neq \Ket{b}$.  Further,
Eq.~\eqref{eq:reproduce} implies
\begin{equation}
  \int_{\Lambda} \xi_M(\psi|\lambda) d\mu_{\psi} = \QProb{\psi}{\psi}
  = 1,
\end{equation}
and in order to satisfy this $\xi_M(\psi|\lambda)$ must be equal to
one on a set that is measure one according to $\mu_{\psi}$.  By
definition, this must be a subset of $\Gamma^{\psi}_M$, so we also
have $\mu_{\psi}(\Gamma^{\psi}_M) = 1$ for any $\Ket{\psi} \in
\mathbb{C}^d$ and any basis $M$ that contains $\Ket{\psi}$.

The goal of this work is to bound the overlap of probability measures
in an ontological model, and this requires a quantitative measure.
For this purpose, define the \emph{classical distance} $D_C$ between
two quantum states in an ontological model to be the variational
distance between the probability measures that represent them, i.e.
\begin{equation}
  D_C(\psi,\phi) = \sup_{\Gamma \in \Sigma} \left [
    \mu_{\psi}(\Gamma) - \mu_{\phi}(\Gamma) \right ].
\end{equation}
This measure has the following operational interpretation.  Suppose a
system is prepared either in the state $\Ket{\psi}$ or the state
$\Ket{\phi}$ with equal a priori probability.  If you are told the
exact value of $\lambda$ then your optimal success probability of
guessing which state was prepared is $\frac{1}{2} \left [ 1 +
  D_C(\psi,\phi) \right ]$.  When $D_C(\psi,\phi) = 1$, the ontic
state $\lambda$ effectively determines which quantum state was
prepared uniquely, so the states have no overlap at the ontic level.
Smaller values of $D_C$ indicate a larger amount of overlap of the
probability measures.  For this reason, it is more convenient to work
with the quantity,
\begin{equation}
  L_C(\psi,\phi) = 1 - D_C(\psi,\phi) = \inf_{\Gamma \in \Sigma}
  \left [ \mu_{\psi}(\Gamma) + \mu_{\phi}(\Lambda \backslash \Gamma)
  \right ],
\end{equation}
which is called the \emph{classical overlap}.

It is important to compare this quantity with a measure of the
indistinguishability of quantum states that has an analogous
interpretation, so that we are comparing like-with-like.  Therefore,
consider the \emph{trace distance} $D_Q$ between quantum states,
which, for pure states, is given by
\begin{equation}
  D_Q(\psi,\phi) = \sqrt{1 - \QProb{\phi}{\psi}}.
\end{equation}
The operational interpretation of this quantity is the same as that of
$D_C$ except that now, instead of being told the ontic state
$\lambda$, you must base your guess as to which state was prepared on
the outcome of a quantum measurement.  In this case, $\frac{1}{2}
\left [ 1 + D_Q(\psi,\phi) \right ]$ is your optimal success
probability for guessing which of $\Ket{\psi}$ or $\Ket{\phi}$ was
prepared.  It is again more convenient to work with the quantity
\begin{equation}
  L_Q(\psi,\phi) = 1 - D_Q(\psi,\phi) = 1 - \sqrt{1 - \QProb{\phi}{\psi}},
\end{equation}
which is called the \emph{quantum overlap}.  

In general $L_C(\psi,\phi) \leq L_Q(\psi,\phi)$.  This is because the
response functions representing measurements provide only coarse
grained information about $\lambda$, so even the optimal quantum
measurement may render $\Ket{\psi}$ and $\Ket{\phi}$ less
distinguishable than they would be if you knew $\lambda$ exactly.
Naively, one might expect that the $\psi$-epistemic explanation of
quantum indistinguishability requires that $L_Q(\psi,\phi) =
L_C(\psi,\phi)$ since then the indistinguishability of $\Ket{\psi}$
and $\Ket{\phi}$ would be entirely accounted for by the classical
indistinguishability of $\mu_{\psi}$ and $\mu_{\phi}$.  However, a
certain amount of coarse-graining of measurements should be expected
in an ontological model.  For example, if the theory is deterministic,
i.e.\ if $\lambda$ determines the outcomes of all measurements
uniquely, then quantum measurements must only reveal coarse grained
information about $\lambda$ on pain of violating the uncertainty
principle.  Therefore, both the overlap of probability measures and
the coarse-grained nature of measurements play a role in explaining
quantum indistinguishability, and one should expect a balance between
these two effects in a viable $\psi$-epistemic theory.  It is only if
$L_C(\psi,\phi) \ll L_Q(\psi,\phi)$ that the $\psi$-epistemic
explanation is in trouble, since then the overlap plays almost no role
in explaining indistinguishability.  For this reason, the scaling of
the ratio $k(\psi,\phi) = L_C(\psi,\phi)/L_Q(\psi,\phi)$, i.e.\ how
quickly it tends to zero in Hilbert space dimension, is of more
interest than its precise value.

The following proposition is the main tool used to bound the classical
overlap.  
\begin{proposition}
  \label{prop:measone}
  Let $\Gamma \in \Sigma$ be a set that is measure one according to
  $\mu_{\phi}$.  Then $L_C(\psi,\phi) \leq \mu_{\psi}(\Gamma)$.
\end{proposition}
\begin{proof}
  Since $\mu_{\phi}(\Gamma) = 1$, $\mu_{\phi}(\Lambda \backslash
  \Gamma) = 0$.  Hence, $L_C(\psi,\phi) = \inf_{\Omega \in \Sigma}
  \left [ \mu_{\psi}(\Omega) - \mu_{\phi}(\Lambda \backslash \Omega)
  \right ] \leq \mu_{\psi}(\Gamma) - \mu_{\phi}(\Lambda \backslash
  \Gamma) = \mu_{\psi}(\Gamma)$.
\end{proof}

A few more definitions are required before proving the main results.
Let $V$ be a finite set of pure states in $\mathbb{C}^d$.  Its
\emph{orthogonality graph} $G = (V,E)$ has the states as its vertices
and there is an edge $\left (\Ket{a},\Ket{b} \right ) \in E$ iff
$\BraKet{a}{b} = 0$.  For every such edge, there exists an orthonormal
basis $M$ such that $\Ket{a},\Ket{b} \in M$.  A \emph{covering set}
$\mathcal{M}$ is a finite set of orthonormal bases such that, for
every $\left (\Ket{a},\Ket{b} \right ) \in E$, there exists an $M \in
\mathcal{M}$ such that $\Ket{a}, \Ket{b} \in M$.  Finally, the
\emph{independence number} $\alpha(G)$ of a graph $G = (V,E)$ is the
cardinality of the largest subset $U \subseteq V$ of vertices such
that if $u \in U$ and $(u,v) \in E$ then $v \notin U$; i.e. $U$
contains no pairs of vertices that are connected by an edge.

\begin{theorem}
  \label{thm:indep}
  Let $V$ be a finite set of pure states in $\mathbb{C}^d$ and let
  $G=(V,E)$ be its orthogonality graph.  Then, for any pure state
  $\Ket{\psi} \in \mathbb{C}^d$, in any ontological model,
  \begin{equation}
    \sum_{\Ket{a} \in V} L_C(\psi,a) \leq \alpha(G).
  \end{equation}
\end{theorem}
\begin{proof}
  Let $\mathcal{M}$ be a covering set of bases.  Then, for $\Ket{a}
  \in V$, define the sets $\Gamma^{a}_{\mathcal{M}} = \cap_{\left \{M
      \in \mathcal{M} \middle |\Ket{a} \in M \right \}} \Gamma^{a}_M$.
  Now, $\Gamma^{a}_{\mathcal{M}}$ is a measure one set according to
  $\mu_{a}$ because it is the intersection of a finite collection of
  measure one sets.  Proposition~\ref{prop:measone} then implies that
  $L_C(\psi,a) \leq \mu_{\psi}(\Gamma^{a}_{\mathcal{M}})$ for any
  $\Ket{\psi}$.  Hence, $\sum_{\Ket{a} \in V} L_C(\psi,a) \leq
  \sum_{\Ket{a} \in V} \mu_{\psi}(\Gamma^{a}_{\mathcal{M}})$.  Now,
  let $\chi_{a}$ be the indicator function of
  $\Gamma^{a}_{\mathcal{M}}$, i.e. $\chi_{a}(\lambda) = 1$ if $\lambda
  \in \Gamma^{a}_{\mathcal{M}}$ and is zero otherwise.  Then,
  \begin{align}
    \sum_{\Ket{a} \in V} \mu_{\psi}(\Gamma^{a}_{\mathcal{M}}) & =
    \sum_{\Ket{a} \in V} \int_{\Lambda}\chi_{a}(\lambda)
    d\mu_{\psi} \\
    & = \int_{\Lambda} \left [ \sum_{\Ket{a} \in V}
      \chi_{a}(\lambda) \right ] d\mu_{\psi} \\
    & \leq \sup_{\lambda \in \Lambda} \left [ \sum_{\Ket{a} \in V}
      \chi_{a}(\lambda) \right ],
  \end{align}
  where the last line follows from convexity.  The last line is upper
  bounded by the maximum number of sets $\Gamma^{a}_{\mathcal{M}}$
  that any given $\lambda$ can be in as $\Ket{a}$ varies over $V$.
  However, $\Gamma^{a}_{\mathcal{M}}$ and $\Gamma^{b}_{\mathcal{M}}$
  are disjoint whenever $\left ( \Ket{a}, \Ket{b} \right ) \in E$
  because they are subsets of $\Gamma^{a}_M$ and $\Gamma^{a}_M$ for
  some basis $M$ and the latter are disjoint.  Therefore, $\lambda$
  can only be in one of $\Gamma^{a}_{\mathcal{M}}$ or
  $\Gamma^{b}_{\mathcal{M}}$ whenever $\Ket{a}$ and $\Ket{b}$ are
  connected by an edge in the orthogonality graph.  Therefore, the
  maximum number of such sets that $\lambda$ can be in is upper
  bounded by the independence number of $G$.
\end{proof}

Readers familiar with the literature on noncontextuality inequalities
will note a similarity between theorem~\ref{thm:indep} and a result of
of \cite{Cabello2010, Cabello2014}, which shows that the maximal
noncontextual value of a class of noncontextuality inequalities is
bounded by the independence number of the orthogonality graph.  This
is not accidental as, up to the removal of measure zero sets, a model
is Kochen-Specker noncontextual if and only if $\int_{\Lambda}
\xi_M(a|\lambda) d\mu_{\psi} = \mu_{\psi}(\Gamma^a_{\mathcal{M}})$,
where now $\mathcal{M}$ is the set of measurement bases involved in
the noncontextuality inequality \cite{Leifer2014}.  Thus, in a
noncontextual model $\mu_{\psi}(\Gamma^a_{\mathcal{M}})$ is the total
probability of obtaining outcome $\Ket{a}$ when measuring a system
prepared in the state $\Ket{\psi}$, so the sum of such probabilities
is bounded in the same way.

\begin{corollary}
  \label{cor:bound}
  Let $V$ be a finite set of pure states in $\mathbb{C}^d$ and let $G
  = (V,E)$ be its orthogonality graph.  For any pure state $\Ket{\psi}
  \in \mathbb{C}^d$ let
  \begin{equation}
    \bar{k}(\psi) = \frac{1}{\Abs{V}} \sum_{\Ket{a} \in V}
    \frac{L_C(\psi,a)}{L_Q(\psi,a)}, 
  \end{equation}
  be the average ratio of classical to quantum overlaps in an
  ontological model.  Then,
  \begin{equation}
    \bar{k}(\psi) \leq \frac{2 \alpha (G)}{\Abs{V} \min_{\Ket{\alpha}
        \in V} \QProb{a}{\psi}}, 
  \end{equation}
  where $\alpha(G)$ is the independence number of $G$.
\end{corollary}
\begin{proof}
  For $0 \leq x \leq 1$, note that $1 - \sqrt{1 - x} \geq
  \frac{x}{2}$, and, hence, $L_Q(\psi,\phi) \geq \frac{1}{2}
  \QProb{\psi}{\phi}$.  Thus,
  \begin{equation*}
    \bar{k}(\psi) \leq \frac{2}{\Abs{V}} \sum_{\Ket{a} \in V}
    \frac{L_C(\psi,a)}{\QProb{\phi}{a}} \leq \frac{2\sum_{\Ket{a} \in
        V} L_C(\psi,a)}{\Abs{V} \min_{\Ket{a} \in V} \QProb{a}{\psi}}.
  \end{equation*}
  The result then follows from theorem~\ref{thm:indep}.
\end{proof}

\begin{theorem}
  When $d$ is divisible by $4$, there exists a set of pure states $V$
  in $\mathbb{C}^d$ and a state $\Ket{\psi} \in \mathbb{C}^d$ such
  that, in any ontological model,
  \begin{equation}
    \bar{k}(\psi) = \sum_{\Ket{a} \in V}
    \frac{L_C(\psi,a)}{L_Q(\psi,a)} \leq de^{-cd}, 
  \end{equation}
  where $c$ is a positive constant.
\end{theorem}
\begin{proof}
  The construction is based on the Hadamard states and the
  Frankl-R{\"odl} theorem, which are commonly used in quantum
  information theory (see e.g.\ \cite{Buhrman1998, Brassard1999,
    Mancinska2013}).  

  Let $V$ be the set of Hadamard states, i.e.\ the set of vectors of
  the form $\frac{1}{\sqrt{d}} \left ( \pm 1, \pm 1, \ldots, \pm 1
  \right )^T$ and let $\Ket{\psi} = \left ( 1, 0, 0, \ldots, 0 \right
  )^T$.  There are $2^d$ vectors in $V$ and each vector $\Ket{a} \in
  V$ satisfies $\QProb{a}{\psi} = \frac{1}{d}$, so this is also the
  minimum.  The orthogonality graph $G = (V,E)$ of $V$ is known in the
  literature as a Hadamard graph \cite{Mancinska2013}.  It follows
  from the Frankl-R{\"o}dl theorem (theorem 1.11 in \cite{Frankl1987})
  that, for $d$ divisible by $4$, there exists an $\epsilon > 0$ such
  that $\alpha(G) \leq (2-\epsilon)^d$ and thus
  corollary~\ref{cor:bound} implies
  \begin{equation}
    \bar{k}(\psi) \leq \frac{2 (2 - \epsilon)^d}{2^d \frac{1}{d}} = 2
    d e^{-cd}, 
  \end{equation}
  where $c = \ln 2 - \ln (2-\epsilon)$ is a positive constant.
\end{proof}

\emph{Remarks.} When $d$ is not divisible by $4$, the Hadamard states
of dimension $\tilde{d} = 4 \left \lfloor \frac{d}{4} \right \rfloor$
can be embedded in $\mathbb{C}^d$ by setting the remaining components
of the vectors to zero.  This yields $\bar{k}(\psi) \leq 2 \tilde{d}
e^{-c \tilde{d}}$ for all Hilbert space dimensions.

Since $\bar{k}(\psi)$ is an average over $V$, there must exist at
least one $\Ket{a} \in V$ such that $k(\psi,a) \leq \bar{k}(\psi) \leq
2 \tilde{d} e^{-c \tilde{d}}$.  It is fairly reasonable to assume that
$k(\psi,a)$ only depends on $\QProb{a}{\psi}$, in which case the bound
$k(\psi,a) \leq 2 \tilde{d} e^{-c \tilde{d}}$ would apply to all
$\Ket{a} \in V$, since $\QProb{a}{\psi} = 1/d$ is the same for the
states used in this construction.

\emph{Conclusions.}  In this letter, I have exhibited a family of
states for which the ratio of classical to quantum overlaps must be
$\leq 2e^{-cd}$ in Hilbert space dimensions $d$ that are divisible by
$4$, and where $c$ is a positive constant.  This represents an
exponential improvement in asymptotic scaling over the previous result
of $4/(d-1)$ \cite{Barrett2013}.  This presents a severe problem for the
$\psi$-epistemic explanation of quantum indistinguishability, as the
portion of the indistinguishability that can be accounted for by the
overlap of probability measures decreases rapidly in large Hilbert
space dimension.  It would be interesting to further pin down the
value of $c$ to see which result gives the best bound for small
Hilbert space dimension, for which the results may be amenable to
experimental test.  Similarly, the connection to contextuality could
be further exploited to derive additional bounds.

Finally, Montina has derived an upper bound on the classical
communication complexity of simulating a qubit identity channel using
the existence of a $\psi$-epistemic ontological model for $d=2$
\cite{Montina2012}.  It would be interesting to determine if lower
bounds on this task in higher dimensions could be derived from overlap
bounds.

\begin{acknowledgments}
  I would like to thank Jonathan Barrett, Owen Maroney, Alberto
  Montina, Simone Severini and Rob Spekkens for useful discussions.
  This research was supported in part by the FQXi Large Grant ``Time
  and the Structure of Quantum Theory'' and by Perimeter Institute for
  Theoretical Physics.  Research at Perimeter Institute is supported
  by the Government of Canada through Industry Canada and by the
  Province of Ontario through the Ministry of Research and Innovation.
\end{acknowledgments}

\bibliography{exp-epistemic}

\end{document}